\documentclass{IEEEtran4PSCC}
\ifCLASSINFOpdf
   \usepackage[pdftex]{graphicx}
\else
   \usepackage[dvips]{graphicx}
\fi
\usepackage[cmex10]{amsmath}
\usepackage{cite}
\usepackage{amsmath}
\usepackage{amsthm}
\usepackage{amssymb}
\usepackage{amsfonts}
\usepackage{graphicx}
\usepackage{dsfont}
\usepackage{siunitx}
\usepackage{multirow}
\usepackage{mathtools}
\usepackage{mathrsfs}
\usepackage{subfigure}
\usepackage[ruled,vlined]{algorithm2e}
\usepackage{algpseudocode}
\usepackage{arydshln}

\usepackage{todonotes}

\DeclareSIUnit[]{\pu}{p.u.}
\DeclareSIUnit[]{\VA}{VA}

\DeclareSymbolFont{bbold}{U}{bbold}{m}{n}
\DeclareSymbolFontAlphabet{\mathbbold}{bbold}

\newcommand{\real}[0]{\mathbb R}

\DeclareSymbolFont{bbold}{U}{bbold}{m}{n}
\DeclareSymbolFontAlphabet{\mathbbold}{bbold}

\newcommand{\diag}[1]{\ensuremath{\mathrm{diag}\left(#1\right)}}

\DeclarePairedDelimiterX\Set[2]{\lbrace}{\rbrace}%
{ #1 \,\delimsize| \,\mathopen{} #2 }

\newtheoremstyle{bfnote}%
{}{}%
{\itshape}{}%
{\bfseries}{.}%
{ }%
{\thmname{#1}\thmnumber{ #2}\thmnote{ (#3)}}
\theoremstyle{bfnote}
\newtheorem{thm}{Theorem}

\newtheorem{rem}{Remark}
\newtheorem{lem}{Lemma}
\newtheorem{ass}{Assumption}
\newtheorem{defn}{Definition}

\usepackage{tikz}
\newcommand*\circled[1]{\tikz[baseline=(char.base)]{\node[shape=circle,draw,inner sep=0.05pt] (char) {#1};}}

\usepackage{enumitem}
\setlist[enumerate]{leftmargin=*}
\setlist[itemize]{leftmargin=*}

\allowdisplaybreaks

\makeatletter
\let\old@ps@headings\ps@headings
\let\old@ps@IEEEtitlepagestyle\ps@IEEEtitlepagestyle
\def\psccfooter#1{%
    \def\ps@headings{%
        \old@ps@headings%
        \def\@oddfoot{\strut\hfill#1\hfill\strut}%
        \def\@evenfoot{\strut\hfill#1\hfill\strut}%
    }%
    \def\ps@IEEEtitlepagestyle{%
        \old@ps@IEEEtitlepagestyle%
        \def\@oddfoot{\strut\hfill#1\hfill\strut}%
        \def\@evenfoot{\strut\hfill#1\hfill\strut}%
    }%
    \ps@headings%
}
\makeatother

\psccfooter{%
        \parbox{\textwidth}{\hrulefill \\ \small{23rd Power Systems Computation Conference} \hfill \begin{minipage}{0.2\textwidth}\centering \vspace*{4pt} \includegraphics[scale=0.06]{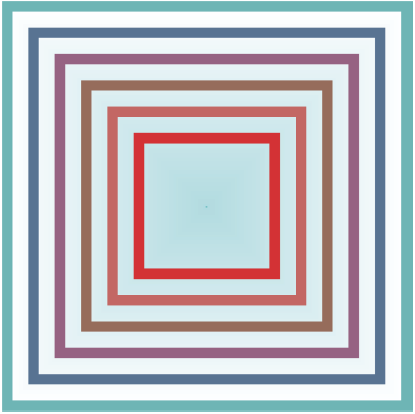}\\\small{PSCC 2024} \end{minipage} \hfill \small{Paris, France --- June 4 -- 7, 2024}}%
}

\begin{document}
%
\title{\LARGE Oscillations-Aware Frequency Security Assessment via Efficient Worst-Case Frequency Nadir Computation}

\author{
\IEEEauthorblockN{Yan Jiang and Baosen Zhang}
\IEEEauthorblockA{Department of Electrical and Computer Engineering \\
University of Washington, Seattle, U.S.A.\\
\{jiangyan, zhangbao\}@uw.edu}
\and
\IEEEauthorblockN{Hancheng Min}
\IEEEauthorblockA{Department of Electrical and Systems Engineering \\
University of Pennsylvania, Philadelphia, U.S.A.\\
hanchmin@seas.upenn.edu}
}


\maketitle

\begin{abstract}
Frequency security assessment following major disturbances has long been one of the central tasks in power system operations. The standard approach is to study the center of inertia frequency, an aggregate signal for an entire system, to avoid analyzing the frequency signal at individual buses. 
However, as the amount of low-inertia renewable resources in a grid increases, the center of inertia frequency is becoming too coarse to provide reliable frequency security assessment. In this paper, we propose an efficient algorithm to determine the worst-case frequency nadir across all buses for bounded power disturbances, as well as identify the power disturbances leading to that severest scenario. The proposed algorithm allows oscillations-aware frequency security assessment without conducting exhaustive simulations and intractable analysis.
\end{abstract}


\thanksto{\noindent Y. Jiang and B. Zhang are partially supported by NSF award ECCS-2153937.}

\section{Introduction}


Frequency security assessment following major disturbances is an important problem in power system operations. In compact power systems where the whole system can be modelled as an equivalent synchronous machine, there is a relatively complete understanding of the frequency dynamics, with explicit formulas derived for quantities such as frequency nadir and rate of change of frequency~\cite{jiangtps2021}. However, for power systems whose inter-machine oscillations are unignorable, the frequency security assessment becomes challenging as the magnitude and location of oscillations depend on a number of factors, including inertia, damping, spatial shape of the excitation, network connectivity, and other operating conditions. 

One approach to reduce the size of the problem is to study the center of inertia (COI) frequency, which is the inertia-weighted average of the frequency signals at individual buses~\cite{milano2017rotor,azizi2020local}. This approach works well when all buses exhibit a relatively coherent frequency response~\cite{Min2021lcss}. However, as increased amount of low-inertia renewable resources located in remote areas enter the grid through long transmission lines, it becomes insufficient to  only consider the COI frequency since significant oscillations can occur which makes the transient frequencies on individual buses deviate drastically from the COI frequency.

Large oscillations can become a major obstacle towards power system stability, which makes them receive extensive attention from the power engineering community~(see~\cite{klein1991fundamental,rogers2000,aboul1996damping,rafique2022bibliographic} and references within). However, the fact that oscillations present a variety of highly-coupled interaction among individual buses makes it nontrivial to provide analytical results in large-scale power systems. For example, although it has been widely realized that the transmission network topology and line parameters play prominent roles in spatiotemporal dynamics of oscillations, their exact impact on system oscillatory behavior is far from well-understood. Therefore, dynamical simulations are typically needed to study the impact of network structure and parameters changes~\cite{rogers2000, Gautam2009TPS, Pagnier2019Plos, zhangtseoscillation}, which is inefficient since the number of possible scenarios that needs to be simulated can be significant before one can gain insight. A recent effort to avoid dynamical simulations is by providing a conservative 
analytical boundary for individual oscillatory bus frequencies under any given power disturbances~\cite{Khamisov2024TPS}, whose cumbersome expressions, however, make it still hard to reveal the system level tolerance to arbitrary bounded power disturbances without running extensive numerical tests.

Recently, techniques and results directly providing analytical insights on oscillations at weak grid conditions have gained renewed interest, which are typically done by investigating the network Laplacian~\cite{guo2018cdc,rouco1998eigenvalue,Minl4dc23cluster} or adjacency~\cite{AllibhoyOJCSoscillation} matrix. These approaches have uncovered many interesting patterns, but they cannot yet answer some specific power system operational questions. For instance, an important question is to find the worst-case oscillations (in terms of frequency nadir) for bounded power disturbances, which can serve as more reliable frequency security assessment in power systems with high renewables than the COI frequency. However, there is no existing guideline on how to do this beyond exhaustive simulations or conservative bounds from spectral graph theory. 

In this paper, we propose an efficient worst-case frequency nadir computation algorithm, which enables us to conduct oscillations-aware frequency security assessment without the need for exhaustive simulations and intractable analysis. Moreover, we demonstrate that what allocation of power disturbances will lead to the worst-case frequency nadir varies with the network connectivity. Particularly, we show the worstness of evenly injected disturbances in homogeneous power networks with strong connectivity. As for more practical settings, one can resort to our proposed algorithm to avoid exhaustive simulations and intractable analysis.

The rest of this paper is organized as follows. Section~\ref{sec:model} describes the power system model and formalizes the worst-case frequency nadir problem. 
Section~\ref{sec:algorithm} derives an efficient algorithm to solve this problem, where a proportionality assumption on generation unit parameters is used to ease the problem.
Section~\ref{sec:homo-case} illustrates, through theoretical analyses on power networks with homogeneous generation units, the high dependency of the worst-case frequency nadir on network connectivity. Section~\ref{sec:simulation} validates the performance of our algorithm through detailed simulations. Section~\ref{sec:conclusion} gathers our conclusions and outlook.


\section{Modeling Approach and Problem Formulation}\label{sec:model}
In this section, we describe the power system model used in this paper and motivate the worst-case frequency nadir problem we aim to solve.
\subsection{Power System Model}
We consider a connected power network composed of $n$ buses indexed by $i \in \mathcal{N} := \{1,\dots, n\} $ and transmission lines denoted by unordered pairs $\{i,j\} \in \mathcal{E}\subset \{\{i,j\}:i,j\in\mathcal{N},i \neq j\}$, whose linearized dynamics around an operating point is shown in Fig.\ref{fig:model}. This is a standard model and interested readers can refer to~\cite{Zhao:2013ts} for details on the linearization procedure. The system is modeled as a feedback interconnection of bus dynamics and network dynamics~\cite{pm2019preprint, jiang2021tac, jiang2021lcss}, where the input and output are the power disturbances $\boldsymbol{p} := \left(p_{i}, i \in \mathcal{N} \right) \in \real^n$ (in $\SI{}{\pu}$) and the bus frequency deviations from the nominal value $\boldsymbol{\omega}:=\left(\omega_i, i \in \mathcal{N} \right) \in \real^n$ (in $\SI{}{\pu}$), respectively.\footnote{Throughout this paper, vectors are denoted in lower case bold and matrices are denoted in upper case bold, while scalars are unbolded, unless otherwise specified. Also, $\mathbbold{1}_n, \mathbbold{0}_n \in \real^n$ denote the vectors of all ones and all zeros, respectively, and $\boldsymbol{e}_k\in\real^n$ denotes the $k$th standard basis vector.} 

\begin{figure}
\centering
\includegraphics[width=\columnwidth]{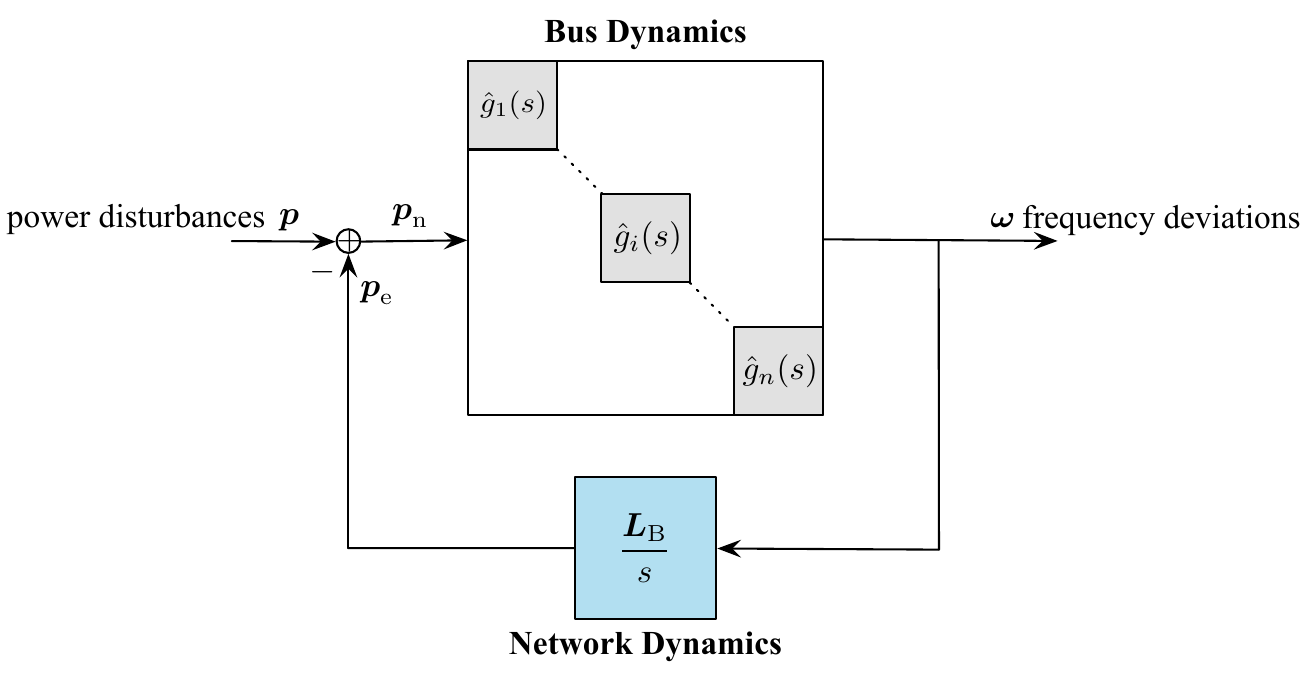}
\caption{Block diagram of power network.}\label{fig:model}
\end{figure}

The bus dynamics that map the net bus power imbalance $\boldsymbol{p}_\mathrm{n} := \left( p_{\mathrm{n},i}, i \in \mathcal{N} \right) \in \real^n$ (in $\SI{}{\pu}$) to the frequency deviations $\boldsymbol{\omega}$ can be described by a diagonal transfer matrix $\hat{\boldsymbol{G}}(s) := \diag{\hat{g}_i(s), i \in \mathcal{N}}$ whose diagonal element $\hat{g}_i(s)$ represents the transfer function of a generation unit on bus $i$, say a synchronous generator or a grid-forming inverter. A standard example typically includes inertia $m_i>0$ (in $\SI{}{\second}$) and damping $d_i>0$ (in $\SI{}{\pu}$) terms, i.e., 
\begin{align}\label{eq:gi-sw}
  \hat{g}_i(s) =\dfrac{1}{m_i s + d_i} \,,\quad \forall i \in \mathcal{N} \,.
\end{align}
Then, we have
\begin{align}\label{eq:bus-dyn}
    \hat {\boldsymbol{\omega}}(s) = \hat{\boldsymbol{G}}(s) \hat{\boldsymbol{p}}_\mathrm{n}(s)\,,
\end{align}
where $\hat {\boldsymbol{\omega}}(s)$ and $\hat {\boldsymbol{p}}_\mathrm{n}(s)$ denote the Laplace transforms of $\boldsymbol{\omega}$ and $\boldsymbol{p}_\mathrm{n}$, respectively.\footnote{We use hat to distinguish the Laplace transform from its time domain counterpart.}

The network dynamics characterizes the relationship between the fluctuations in power drained into the transmission
network $\boldsymbol{p}_\mathrm{e} := \left(p_{\mathrm{e},i}, i \in \mathcal{N} \right) \in \real^n$ (in $\SI{}{\pu}$) and the frequency deviations $\boldsymbol{\omega}$, which is given by a linearized model of the power flow equations~\cite{Purchala2005dc-flow}:
\begin{align}
 \hat {\boldsymbol{p}}_\mathrm{e}(s) = \frac{\boldsymbol{L}_\mathrm{B}}{s} \hat {\boldsymbol{\omega}}(s)\;,\label{eq:Network-dyn}
\end{align}
 where the matrix $\boldsymbol{L}_\mathrm{B}:=\left[ L_{\mathrm{B},{ij}}\right]\in \real^{n \times n}$ is an undirected weighted Laplacian matrix of the network whose $ij$th element is 
\[
\boldsymbol{L}_{\mathrm{B},{ij}}=\Omega_0\partial_{\theta_j}{\sum_{l=1}^n|V_i||V_l|B_{il}\sin(\theta_i-\theta_l)}\Bigr|_{\boldsymbol{\theta}=\boldsymbol{\theta}_0}.
\]
Here, $\boldsymbol{\theta} := \left(\theta_i, i \in \mathcal{N} \right) \in \real^n$ are the voltage angles with $\boldsymbol{\theta}_0$ being the equilibrium angles (in $\SI{}{\radian}$), $|V_i|$ is the (constant) voltage magnitude at bus $i$ (in $\SI{}{\pu}$), $B_{ij}$ is the line $\{i,j\}$ susceptance (in $\SI{}{\pu}$), and $\Omega_0:=2\pi F_0$ is the nominal frequency (in $\SI{}{\radian/s}$) with $F_0$ being $\SI{50}{\hertz}$ or $\SI{60}{\hertz}$
depending on the particular system. 

We are interested in the closed-loop response of frequency deviations $\boldsymbol{\omega}$ following power disturbances $\boldsymbol{p}$ in the power network shown in Fig.~\ref{fig:model}. This can be obtained by combining \eqref{eq:bus-dyn} and \eqref{eq:Network-dyn} through the relation $\boldsymbol{p}_\mathrm{n}=\boldsymbol{p}-\boldsymbol{p}_\mathrm{e}$ as 
\begin{equation}\label{eq:twp}
\hat {\boldsymbol{\omega}}(s)= \left(\boldsymbol{I}_n +\hat{\boldsymbol{G}}(s)\frac{\boldsymbol{L}_\mathrm{B}}{s}\right)^{-1} \hat{\boldsymbol{G}}(s)\hat{\boldsymbol{p}}(s)\,. 
\end{equation}
However, although \eqref{eq:twp} is a closed-form expression, it is difficult to work with since the size of the matrices and vectors involved could be quite high. A standard simplification is to work with the center of inertia (COI) frequency, defined as 
\begin{equation}\label{eq:COI}
    \bar{\omega}:= \dfrac{\sum_{i=1}^n m_i\omega_i}{\sum_{i=1}^n m_i}\,,
\end{equation}
which is the inertia-weighted average
of individual bus frequencies. The COI frequency is a good representative of the system frequency response if the power network is tightly-connected~\cite{Min2021lcss}. However, as more renewable resources are integrated, especially at edges of the power network, significant oscillatory behavior could occur. Thus, it becomes insufficient to only consider the COI frequency. 

\subsection{Illustrative Example of Frequency Oscillations}\label{ssec:example}
\begin{figure}[t!]
\centering
\subfigure[All buses are tightly-connected]
{\includegraphics[width=0.35\columnwidth]{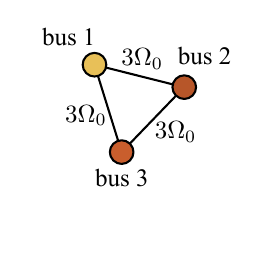}\includegraphics[width=0.6\columnwidth]{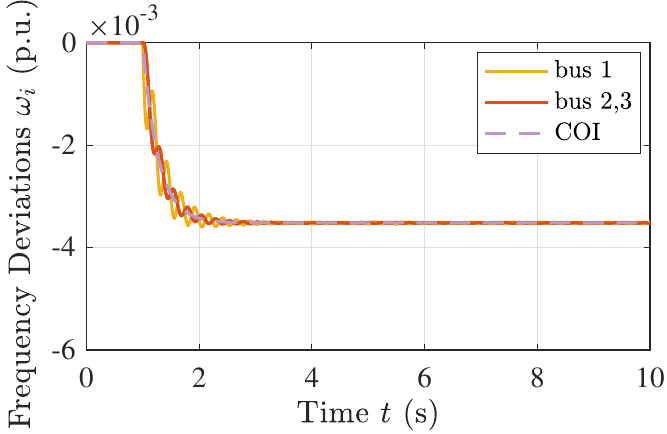}\label{fig:fre-tight-3bus}}
\subfigure[Bus $1$ is weakly-connected from buses $2$ and $3$]
{\includegraphics[width=0.35\columnwidth]{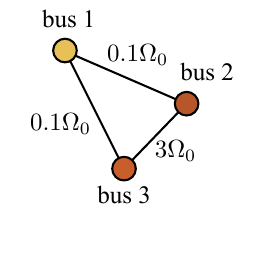}
\includegraphics[width=0.6\columnwidth]{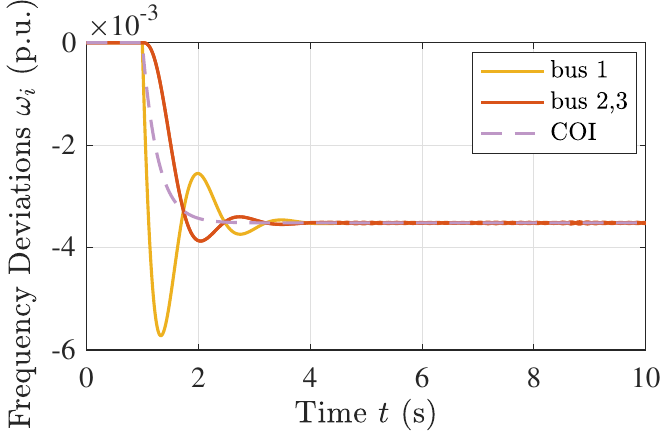}
\label{fig:fre-weak-3bus}}
\caption{A $3$-bus example showing that COI frequency is insufficient for frequency security assessment since the oscillations due to weak connectivity can make the transient frequencies on individual buses deviate drastically from the COI frequency.}
\label{fig:3-bus}
\vspace{-0.4cm}
\end{figure}

To see this, we provide a toy $3$-bus power network example in Fig.~\ref{fig:3-bus} to show different types of oscillatory behaviors. Basically, we compare the frequency responses of the system with strong and weak network connectivity by varying the line impedance parameters on the same network topology, where the same generation units are deployed and the same power disturbances are applied in two cases for a fair comparison.\footnote{Parameters of generation units are chosen to emulate the Great Britain power system under the high renewable penetration scenario~\cite[Table 1]{jiangtps2021}, i.e., $m_i=\SI{4.38}{\second}$ and $d_i=\SI{16}{\second}$ for all buses.} 
When
sudden power disturbances $\boldsymbol{p}= [-0.1689, 0, 0]^T \mathds{U}_{ t \geq 0 }$~$\SI{}{\pu}$ occur at $t=\SI{1}{\second}$,\footnote{$\mathds{U}_{ t \geq 0 } $ denotes the unit-step function.} the trajectories of the frequencies are as shown in Fig.~\ref{fig:3-bus}. 

For the more tightly-connected system in Fig.~\ref{fig:fre-tight-3bus}, the frequency trajectories at all buses closely track each other, even though the disturbance is only at bus $1$. In this case, the COI frequency acts as a good proxy of the nodal frequencies. In contrast, the frequency trajectories in Fig.~\ref{fig:fre-weak-3bus} exhibit significant oscillations against each other with large nadirs while the COI frequency is well-behaved. Therefore, relying just on the COI frequency could lead to an erroneous conclusion about the frequency security of the system. Of course, the behavior in Fig.~\ref{fig:fre-weak-3bus} is not unexpected since bus $1$ is ``weakly'' connected with the other buses. However, as the power network gets larger, it is more challenging to draw intuitive conclusions, which makes additional analysis techniques become necessary.

\subsection{Oscillations-Aware Frequency Security Assessment}\label{sec:opt}
To provide quantitative tools for frequency security assessment, we focus on the metric of frequency nadir in the rest of the paper, which is defined as the maximum frequency deviation from the nominal frequency at each bus during a transient response, i.e., $\max_{t\geq0} |\omega_i(t)|$. Since a too large frequency nadir could trigger undesired protection measures and even cause cascading failures, its value is crucial to frequency security. For example, the maximum allowed frequency nadir is $\SI{800}{\milli\hertz}$ ($\SI{0.016}{\pu}$ on a $\SI{50}{\hertz}$ base) for the European system \cite{Knap2016size} and $\SI{500}{\milli\hertz}$ ($\SI{0.01}{\pu}$ on a $\SI{50}{\hertz}$ base) for the Great Britain system \cite{NG2016,NG2016standard}. 

To make a power network more reliable, our first order of business involves securing its frequency response against unanticipated power disturbances. That is, we would like to prejudge whether the frequency nadir on each bus would stay in the allowed threshold for any step power disturbances $\boldsymbol{p}= \boldsymbol{u}_0 \mathds{U}_{ t \geq 0 }$ with $\boldsymbol{u}_0 \in \real^n$ being arbitrary up to a norm constraint $\|\boldsymbol{u}_0\|\leq\rho$ for some $\rho>0$. Compactly, this asks for the solution of the following \emph{worst-case frequency nadir problem}:
\begin{align}\label{eq:opt-worst-nadir}
	\max_{ \substack{\boldsymbol{p}=\boldsymbol{u}_0 \mathds{U}_{ t \geq 0 }\\\Set*{\boldsymbol{u}_0 \in \real^n\!}{\!\|\boldsymbol{u}_0\|\leq\rho}}} \!\!\!\!\!\! & \max_{i \in \mathcal{N}}\ \  \max_{t>0} |\omega_i(t)|
   & \mathrm{s.t.}\quad  
		\mathrm{dynamics}\ \eqref{eq:twp}\,,
  \end{align}
which is computationally challenging due to the nested maximums over all power disturbance profiles, all buses, and all time.   

Our goal is to design an efficient algorithm to solve the above problem. This is in general challenging due to the high complexity of frequency dynamics. Nevertheless, when the parameters of generation units satisfy a commonly used proportionality assumption~\cite{guo2018cdc, pm2019preprint, jiang2021tac}, the frequency responses can be decomposed in a way that allows us to solve \eqref{eq:opt-worst-nadir} efficiently. Moreover, for certain types of networks, \eqref{eq:opt-worst-nadir} can be solved analytically, providing interesting intuition about the  spatial and temporal behavior of oscillations. 

\section{Efficient Worst-Case Frequency Nadir Algorithm for Proportional Power Networks}\label{sec:algorithm}
In this section, we show that, under a simplifying assumption, it is possible to decompose the frequency dynamics, which enables us to propose an efficient algorithm to solve the worst-case frequency nadir problem formulated in Section~\ref{sec:opt}.
\subsection{Modal Decomposition of Frequency Dynamics}\label{sec:modal}
The example in Section~\ref{ssec:example} shows that the frequency nadir of a power network is a result of interference between the frequency oscillations at each bus, which can have a complicated dependence on the parameters and topology of the network. Hence, a better evaluation of the worst-case frequency nadir requires a deeper understanding of frequency oscillations. However, since the oscillations present a variety of highly-coupled interactions between individual buses, it is challenging to solve \eqref{eq:opt-worst-nadir} in large-scale power systems with heterogeneous generation units. 
To make the analysis tractable, we consider proportionality as a reasonable
first-cut approximation to heterogeneity~\cite{pm2019preprint}, under which the frequency dynamics \eqref{eq:twp} are decoupled and thus the problem \eqref{eq:opt-worst-nadir} can be solved efficiently. Hence, we adopt the following assumption in the rest of this paper:
\begin{ass}[Proportionality]\label{ass:proportion}
There exists a group of proportional parameters $r_i>0$, $i \in \mathcal{N}$, such that 
\[ \hat{g}_i(s) = \dfrac{\hat{g}_\mathrm{o}(s)}{r_i} \,, \]
where $\hat{g}_\mathrm{o}(s)$ is called the representative generation unit.
\end{ass}

For example, for the case where $\hat{g}_i(s)$ has \eqref{eq:gi-sw} as its model, Assumption~\ref{ass:proportion} is satisfied provided that inertia and
damping are both proportional to $r_i$. That is, $\forall i \in \mathcal{N}$, $\exists r_i>0$ such that $m_i=r_im$ and $d_i=r_id$ for some  $m$ and $d$, i.e., 
\begin{equation}\label{eq:go}
\hat{g}_\mathrm{o}(s)=\dfrac{1}{ms+d} \,.   
\end{equation}
\begin{rem}[Proportionality extensions]
The practical relevance of Assumption~\ref{ass:proportion} is justified in many empirical studies. For example, \cite{oakridge2013} shows that, at least in regards of order of magnitude, Assumption \ref{ass:proportion} is a reasonable first-cut approximation to heterogeneity.
In fact, the heterogeneous case can be considered as a diagonal perturbation from the ideal proportional case \cite[Section VI-A]{pm2019preprint}. Thus, the results derived from proportional power networks can be extended to networks with heterogeneous generation units. In this case, the performance can be assessed through a robustness analysis, which is an interesting direction for further research.
\end{rem}

Under Assumption~\ref{ass:proportion}, it has been well-established that the dynamics in \eqref{eq:twp} can be decoupled as~\cite{pm2019preprint, jiang2021tac}
\begin{align}
\hat {\boldsymbol{\omega}}(s) = \boldsymbol{R}^{-\frac{1}{2}} \boldsymbol{V} \diag{\hat{z}_{k}(s), k \in \mathcal{N}} \boldsymbol{V}^T \boldsymbol{R}^{-\frac{1}{2}}\hat{\boldsymbol{p}}(s)\;,\label{eq:Tp}
\end{align}
with
\begin{align}
    \hat{z}_{k}(s) := \frac{s\hat{g}_\mathrm{o}(s)}{s+\lambda_k\hat{g}_\mathrm{o}(s)}\;,\quad \forall k \in \mathcal{N}\,, \label{eq:zk-s}
\end{align}
where $\boldsymbol{R} :=\diag {r_i, i \in \mathcal{N}} \in \real^{n \times n}$ is the proportionality matrix and
\begin{align}\label{eq:Vmatrix}
\boldsymbol{V} \!:=\! \begin{bmatrix} \boldsymbol{v}_1\!:=\!(\sum_{i=1}^n \!r_i)^{-\frac{1}{2}} \boldsymbol{R}^{\frac{1}{2}} \mathbbold{1}_n \!& \!\boldsymbol{v}_2 &\!\cdots \!& \boldsymbol{v}_n \end{bmatrix}\!\!\in\! \real^{n \times n}    
\end{align}
satisfying $\boldsymbol{V}^T\boldsymbol{V}=\boldsymbol{V}\boldsymbol{V}^T=\boldsymbol{I}_n $ is an orthonormal matrix whose columns $\boldsymbol{v}_k:=\left(v_{k,i}, i \in \mathcal{N} \right) \in \real^n$ are unit eigenvectors associated with the scaled Laplacian matrix 
\begin{align}\label{eq:L-def}
\boldsymbol{L} := \boldsymbol{R}^{-\frac{1}{2}} \boldsymbol{L}_\mathrm{B} \boldsymbol{R}^{-\frac{1}{2}}\in \real^{n \times n}    
\end{align} 
such that $\boldsymbol{L} = \boldsymbol{V} \diag{\lambda_k, k \in \mathcal{N}} \boldsymbol{V}^T$
with $\lambda_k$ being the $k$th eigenvalue of $\boldsymbol{L}$ ordered non-decreasingly $(0 = \lambda_1 < \lambda_2 \leq \ldots \leq\lambda_n)$\footnote{Recall that we assume the power network is connected, which means that $\boldsymbol{L}$ has a single zero eigenvalue.}.

Note that $\boldsymbol{V}$ has orthonormal columns $\boldsymbol{v}_k$, $\forall k \in \mathcal{N}$, which are automatically linearly independent~\cite[Theorem 2.1.2]{Horn2012MA}. 
It readily follows that $\boldsymbol{R}^{\frac{1}{2}}\boldsymbol{v}_k$, $\forall k \in \mathcal{N}$, are linearly independent as well since $\boldsymbol{R}$ is nonsingular by its construction. Thus, the $n$ linearly independent vectors $\boldsymbol{R}^{\frac{1}{2}}\boldsymbol{v}_k$, $\forall k \in \mathcal{N}$, form a basis that spans $\real^n$. This implies that any step power disturbances $\boldsymbol{p}=\boldsymbol{u}_0 \mathds{U}_{ t \geq 0 }$ that the power network undergoes can always be decomposed along this basis, i.e., $\exists\boldsymbol{\alpha}:=\left(\alpha_i, i \in \mathcal{N} \right) \in \real^n$ such that
\begin{align}\label{eq:u0-vk}
    \boldsymbol{u}_0=\sum_{k=1}^n \alpha_k \boldsymbol{R}^{\frac{1}{2}}\boldsymbol{v}_k=\boldsymbol{R}^{\frac{1}{2}}\boldsymbol{V}\boldsymbol{\alpha}\,.
\end{align}
Combining \eqref{eq:Tp} and \eqref{eq:u0-vk} via $\hat{\boldsymbol{p}}(s)=\boldsymbol{u}_0/s$ yields the following decomposition of frequency responses along the scaled Laplacian eigenvectors.

\begin{lem}[Decomposition along scaled Laplacian eigenvectors] \label{lem:w-t}Under Assumption~\ref{ass:proportion}, if the power network in Fig.~\ref{fig:model} undergoes step power disturbances $\boldsymbol{p}=\boldsymbol{u}_0 \mathds{U}_{ t \geq 0 }$, then the frequency responses can be decoupled into
   \begin{align}
    {\boldsymbol{\omega}}(t)=&\sum_{k=1}^n\alpha_k h_k(t)\boldsymbol{R}^{-\frac{1}{2}}\boldsymbol{v}_k\label{eq:omega-t-decomple}\\=&\ \dfrac{\alpha_1h_1(t)}{\sqrt{\sum_{i=1}^n r_i}}\mathbbold{1}_n+\sum_{k=2}^n\alpha_k h_k(t)\boldsymbol{R}^{-\frac{1}{2}}\boldsymbol{v}_k\,,\label{eq:omega-t-decomple-coi}
\end{align} 
where
\begin{align}
\boldsymbol{\alpha}:=&\left(\alpha_i, i \in \mathcal{N} \right)=\boldsymbol{V}^T\boldsymbol{R}^{-\frac{1}{2}}\boldsymbol{u}_0\,,\label{eq:alpha}\\
    h_k(t):=&\ \mathscr{L}^{-1}\left\{\frac{\hat{g}_\mathrm{o}(s)}{s+\lambda_k\hat{g}_\mathrm{o}(s)}\right\}\;,\quad \forall k \in \mathcal{N}\,.\label{eq:hk-def}
\end{align}
\end{lem}

\begin{proof}
See the Appendix \ref{app:lem1-pf}.
\end{proof}

Lemma~\ref{lem:w-t} makes the physical interpretation of modal decomposition more clear. We can observe from \eqref{eq:omega-t-decomple} that the frequency responses are a linear combination of independent modes along the directions dictated by the scaled Laplacian eigenvectors. For example, $h_k(t)$ captures the frequency response of the power network along $\boldsymbol{R}^{-\frac{1}{2}}\boldsymbol{v}_k$ to step power disturbances. Particularly, as shown in \eqref{eq:omega-t-decomple-coi}, the first mode actually characterizes the common behavior among individual buses as a scaled step response of the representative generation unit $\hat{g}_\mathrm{o}(s)$ while the remaining modes represent the oscillations among them. In fact, one can easily show by a similar argument as in \cite{pm2019preprint} that the common behavior is the $r_i$-weighted average of individual bus frequencies, i.e., $(\sum_{i=1}^n r_i\omega_i)/(\sum_{i=1}^n r_i)$, which is exactly the COI frequency defined in \eqref{eq:COI} if $\hat{g}_i(s)$ has \eqref{eq:gi-sw} as its model. Therefore, the modal decomposition in \eqref{eq:omega-t-decomple-coi} confirms our intuition that the COI frequency becomes insufficient for frequency security assessment in a power network where those oscillatory modes are unignorable.

\subsection{Worst-Case Frequency Nadir Algorithm}\label{ssec:alg}
The decoupled frequency dynamics \eqref{eq:omega-t-decomple} derived in Section~\ref{sec:modal} enables us to propose an efficient algorithm summarized in Algorithm~\ref{alg:Nadir} to solve the worst-case frequency nadir problem \eqref{eq:opt-worst-nadir} in Section~\ref{sec:opt}. We now explain the rationale behind Algorithm~\ref{alg:Nadir} in detail.

\makeatletter
\patchcmd{\@algocf@start}
  {-1.5em}
  {0pt}
  {}{}
\makeatother
\DecMargin{0.2cm} 
\SetArgSty{textnormal}
\begin{algorithm}
\caption{Worst-Case Frequency Nadir Computation}\label{alg:Nadir}
    \KwData{$\rho$, $\boldsymbol{R}$, $\boldsymbol{V}$, $\boldsymbol{c}_i(t):=\left(v_{k,i}h_k(t), k \in \mathcal{N} \right)$}
    \KwResult{optimal solution to the problem \eqref{eq:opt-worst-nadir}}
    Choose: length $\delta t>0$ and number $N>0$ of time steps\;
    Initialize: storing table $\boldsymbol{F}:=\left[F_{it}\right]\in \real^{n \times N}\gets\boldsymbol{0}_{n\times N}$\;
    \For{$i\in\mathcal{N}$}{
    \For{$t\in\mathcal{T}:= \{\delta t,\dots, N\delta t\}$}{
    $F_{it}\gets\dfrac{\rho}{\sqrt{r_i}}\left\lVert\boldsymbol{R}^{-\frac{1}{2}}\boldsymbol{V}\boldsymbol{c}_i(t) \right\rVert^\mathrm{D}$\;
    }
    }
    $(i^\star,t^\star)\gets\arg \max_{(i,t)\in\mathcal{N}\times\mathcal{T}} F_{it}$\;
    \tcc{Return worst-case frequency nadir and corresponding power disturbances}
    \Return{$(i^\star,t^\star)$, $F_{i^\star t^\star}$ and $\boldsymbol{u}_{0}^\star(i^\star,t^\star)$}\;
\end{algorithm}

Firstly, with the aid of Lemma~\ref{lem:w-t}, we can reformulate the problem \eqref{eq:opt-worst-nadir} into
\begin{equation}\label{eq:opt-worst-nadir-decouple}
	\max_{ \substack{\Set*{\boldsymbol{u}_0\in \real^n\!}{\!\|\boldsymbol{u}_0\|\leq\rho}}} \max_{i \in \mathcal{N}}\ \  \max_{t>0} \left\lvert\sum_{k=1}^n\alpha_k h_k(t)\dfrac{v_{k,i}}{\sqrt{r_i}}\right\rvert \,, 
\end{equation}
where the objective function is simply the magnitude of the $i$th element of \eqref{eq:omega-t-decomple}. Note that $\alpha_k$ in \eqref{eq:opt-worst-nadir-decouple} actually depends on the optimization variable $\boldsymbol{u}_0$ via \eqref{eq:alpha}. With this in mind, we can further turn \eqref{eq:opt-worst-nadir-decouple} into the following equivalent problem
\begin{equation}\label{eq:opt-worst-nadir-decouple-swap}
	\max_{i \in \mathcal{N}}  \max_{t>0} \max_{ \substack{\Set*{\boldsymbol{u}_0\in \real^n\!}{\!\|\boldsymbol{u}_0\|\leq\rho}}}\dfrac{1}{\sqrt{r_i}}\left\lvert\boldsymbol{c}_i(t)^T\boldsymbol{V}^T\boldsymbol{R}^{-\frac{1}{2}}\boldsymbol{u}_0 \right\rvert\,,
\end{equation}
where we define $\boldsymbol{c}_i(t):=\left(v_{k,i}h_k(t), k \in \mathcal{N} \right) \in \real^n$, swap the order of maximization, and use \eqref{eq:alpha}. Moreover, as we will show later, \eqref{eq:opt-worst-nadir-decouple-swap} can be further simplified as
\begin{align}\label{eq:cen-opt}
   \max_{i \in \mathcal{N}} \max_{t>0}\dfrac{\rho}{\sqrt{r_i}}\left\lVert\boldsymbol{R}^{-\frac{1}{2}}\boldsymbol{V}\boldsymbol{c}_i(t) \right\rVert^\mathrm{D}
\end{align}
whose objective function is the closed-form optimal solution to the inner most optimization problem in \eqref{eq:opt-worst-nadir-decouple-swap}, with the superscript ``$\mathrm{D}$'' denoting the dual norm defined as follows:
\begin{defn}[Dual norm~\cite{Horn2012MA}]
    Given a norm $\|\cdot\|$ on $\real^n$, its dual norm $\|\cdot\|^\mathrm{D}$ is the function from $\real^n$ to $\real$ with values 
\begin{equation}\label{eq:dual-1}
    \|\boldsymbol{x}\|^\mathrm{D}:=\max_{ \substack{\Set*{\boldsymbol{y}:=\left(y_{i}, i \in \mathcal{N} \right) \in \real^n\!}{\!\|\boldsymbol{y}\|\leq1}}}\left\lvert\boldsymbol{x}^T\boldsymbol{y}\right\rvert\,.
\end{equation}     
\end{defn}

The explicit expression of dual norm depends on the type of the norm adopted for bounding $\boldsymbol{y}$ in \eqref{eq:dual-1}~\cite[Equation 5.4.15a]{Horn2012MA}. The dual norms of some common norms are provided by the following lemma.

\begin{lem}[Solutions to dual norm]\label{lem:dual-norm}  
$\forall\boldsymbol{x}:=\left(x_{i}, i \in \mathcal{N} \right)\neq\mathbbold{0}_n \in \real^n$, the solution to \eqref{eq:dual-1} for the case when $\|\cdot\|$ is: 
\begin{itemize}
\item ($2$-norm) $\|\boldsymbol{y}\|_2:=\sqrt{\sum_{i=1}^n y_i^2}$ is $\|\boldsymbol{x}\|_2^\mathrm{D}=\|\boldsymbol{x}\|_2$, which is achieved at $\boldsymbol{y}^\star=\pm \boldsymbol{x}/\|\boldsymbol{x}\|_2$;
\item ($\infty$-norm) $\|\boldsymbol{y}\|_\infty:=\max_{i\in\mathcal{N}}|y_i|$ is $\|\boldsymbol{x}\|_\infty^\mathrm{D}=\|\boldsymbol{x}\|_1$, which is achieved at $\boldsymbol{y}^\star=\pm \mathrm{sign}(\boldsymbol{x})$;
\item ($1$-norm) $\|\boldsymbol{y}\|_1:=\sum_{i=1}^n|y_i|$ is $\|\boldsymbol{x}\|_1^\mathrm{D}=\|\boldsymbol{x}\|_\infty$, which is achieved at $\boldsymbol{y}^\star=\pm \boldsymbol{e}_{i_\mathrm{m}}$ with $i_\mathrm{m}:=\arg\max_{i\in\mathcal{N}}|x_i|$ .
\end{itemize}
\end{lem}
\begin{proof}
See the Appendix \ref{app:lem2-pf}.      
\end{proof}

In fact, deriving the objective function of \eqref{eq:cen-opt} is simply a matter of applying Lemma~\ref{lem:dual-norm} to the inner most optimization problem in \eqref{eq:opt-worst-nadir-decouple-swap} such that closed-form optimal solutions are obtained under different types of norm bounds on $\boldsymbol{u}_0$.
\begin{thm}[Closed-form solution to the inner most problem in \eqref{eq:opt-worst-nadir-decouple-swap}] \label{them:in-opt}Consider the inner most problem in \eqref{eq:opt-worst-nadir-decouple-swap}, i.e.,
\begin{align}\label{eq:in-pro}
    \max_{ \substack{\Set*{\boldsymbol{u}_0\in \real^n\!}{\!\|\boldsymbol{u}_0\|\leq\rho}}} \dfrac{1}{\sqrt{r_i}}\left\lvert\left(\boldsymbol{R}^{-\frac{1}{2}}\boldsymbol{V}\boldsymbol{c}_i(t)\right)^T\boldsymbol{u}_0 \right\rvert\,.
\end{align}
If the norm constraint on $\boldsymbol{u}_0$ is:
\begin{itemize}
\item ($2$-norm) $\|\boldsymbol{u}_0\|_2\leq\rho$, then the maximum value of \eqref{eq:in-pro} is
\begin{align}\label{eq:sol-2norm}
    \dfrac{\rho}{\sqrt{r_i}}\left\lVert\boldsymbol{R}^{-\frac{1}{2}}\boldsymbol{V}\boldsymbol{c}_i(t) \right\rVert _2=:\dfrac{\rho}{\sqrt{r_i}}\left\lVert\boldsymbol{R}^{-\frac{1}{2}}\boldsymbol{V}\boldsymbol{c}_i(t) \right\rVert _2^\mathrm{D}\,,
\end{align}
achieved at $\boldsymbol{u}_{0}^\star(i,t) = \pm \rho\boldsymbol{R}^{-\frac{1}{2}}\boldsymbol{V}\boldsymbol{c}_i(t)/\|\boldsymbol{R}^{-\frac{1}{2}}\boldsymbol{V}\boldsymbol{c}_i(t) \|_2$;
\item ($\infty$-norm) $\|\boldsymbol{u}_0\|_\infty\leq\rho$, then the maximum value of \eqref{eq:in-pro} is
\begin{align}\label{eq:sol-inf}
    \dfrac{\rho}{\sqrt{r_i}}\left\lVert\boldsymbol{R}^{-\frac{1}{2}}\boldsymbol{V}\boldsymbol{c}_i(t) \right\rVert _1=:\dfrac{\rho}{\sqrt{r_i}}\left\lVert\boldsymbol{R}^{-\frac{1}{2}}\boldsymbol{V}\boldsymbol{c}_i(t) \right\rVert _\infty^\mathrm{D}\,,
\end{align}
achieved at $\boldsymbol{u}_{0}^\star(i,t) = \pm\rho \mathrm{sign}(\boldsymbol{R}^{-\frac{1}{2}}\boldsymbol{V}\boldsymbol{c}_i(t))$;
\item ($1$-norm) $\|\boldsymbol{u}_0\|_1\leq\rho$, then the maximum value of \eqref{eq:in-pro} is
\begin{align}\label{eq:sol-1norm}
    \dfrac{\rho}{\sqrt{r_i}}\left\lVert\boldsymbol{R}^{-\frac{1}{2}}\boldsymbol{V}\boldsymbol{c}_i(t) \right\rVert _\infty=:\dfrac{\rho}{\sqrt{r_i}}\left\lVert\boldsymbol{R}^{-\frac{1}{2}}\boldsymbol{V}\boldsymbol{c}_i(t) \right\rVert _1^\mathrm{D}\,,
\end{align}
achieved at $\boldsymbol{u}_{0}^\star(i,t) = \pm\rho \boldsymbol{e}_{i_\mathrm{m}}$ with $i_\mathrm{m}\!:=\!\arg\max_{i\in\mathcal{N}}|\boldsymbol{e}_i^T\boldsymbol{R}^{-\frac{1}{2}}\boldsymbol{V}\boldsymbol{c}_i(t)|$.
\end{itemize}
\end{thm}
\begin{proof} 
Besides the trivial effect from the scale factor $1/\sqrt{r_i}$ and $\rho$ on the optimal value, the results follow directly from the same argument as in the proof of Lemma~\ref{lem:dual-norm} by setting $\boldsymbol{x}=\boldsymbol{R}^{-\frac{1}{2}}\boldsymbol{V}\boldsymbol{c}_i(t)$ and $\boldsymbol{y}=\boldsymbol{u}_{0}$.
\end{proof}

Theorem~\ref{them:in-opt} allows us to simplify \eqref{eq:opt-worst-nadir-decouple-swap} into \eqref{eq:cen-opt} with the explicit expression of the dual norm being provided by \eqref{eq:sol-2norm}, \eqref{eq:sol-inf}, and \eqref{eq:sol-1norm}, respectively, for the case when the constraint on $\boldsymbol{u}_{0}$ is enforced through $2$-norm, $\infty$-norm, and $1$-norm, respectively. 

Now, it is easy to see that the framework of Algorithm~\ref{alg:Nadir} results from \eqref{eq:cen-opt}. Basically, we simply evaluate its objective function for any $i \in \mathcal{N}$ and $t\in\mathcal{T}:= \{\delta t,\dots, N\delta t\}$ under the chosen length $\delta t>0$ and number $N>0$ of time steps. This yields a $n$ by $N$ table from which we can easily read off the value of the worst-case frequency nadir and $(i^\star,t^\star)$ leading to that value. We highlight that this procedure can be done efficiently in that the dual norm required in \eqref{eq:cen-opt} is cheap to compute with its closed-from expression provided by Theorem~\ref{them:in-opt} under common norm constraints on $\boldsymbol{u}_0$. Once $(i^\star,t^\star)$ is determined, we can easily identify the power disturbances leading to that severest scenario as $\boldsymbol{u}_{0}^\star(i^\star,t^\star)$ whose explicit form is provided by Theorem~\ref{them:in-opt} as well.

\section{Worst-Case Frequency Nadir of Power Networks with Homogeneous Generation Units}\label{sec:homo-case}
In general, even under the proportionality assumption,  it is difficult to provide an analytical form for the worst-case frequency nadir. Therefore, in this section, we focus on the case with homogeneous generation units. That is, we assume $\boldsymbol{R}=\boldsymbol{I}_n$, which implies that $m_i=m$ and $d_i=d$, $\forall i\in \mathcal{N}$, if the representative generation unit is given by \eqref{eq:go}.
This type of network serves as an interesting example to illustrate different phenomena that can occur in power systems with various connectivity.  The next lemma is helpful for the analysis. 


\begin{lem}[Expressions of $h_k(t)$]\label{lem:hk-exp}Under Assumption~\ref{ass:proportion}, if the power network in Fig.~\ref{fig:model} has its representative generation unit given by \eqref{eq:go}, then 
   \begin{align}
    h_1(t)=\dfrac{1}{d}\left(1-e^{-\frac{d}{m}t}\right)\label{eq:h1-def}
\end{align}
and $h_k(t)$, $\forall k\in\mathcal{N}\setminus\{1\}$, falls into one of the three forms depending on the value of damping ratio $\xi_k$:
\begin{align*}
    h_k(t)\!=\!\left\{
    \begin{aligned}
    &\dfrac{e^{-\xi_k\omega_{\mathrm{n}k} t}}{m\omega_{\mathrm{n}k} \sqrt{1 - \xi_k^2}}\sin{\!\left(\!\omega_{\mathrm{n}k} \sqrt{1 - \xi_k^2}t\!\right)\!}\ \ \mathrm{if}\ 0\leq\xi_k<1\mathrm{,}\\ 
    &\dfrac{1}{m}e^{-\omega_{\mathrm{n}k} t}t\ \ \mathrm{if}\ \xi_k=1\mathrm{,}\\ 
    &\dfrac{e^{-\xi_k\omega_{\mathrm{n}k} t}\!\left(e^{\omega_{\mathrm{n}k}\sqrt{\xi_k^2-1} t}\!-\! e^{-\omega_{\mathrm{n}k} \sqrt{\xi_k^2-1}t}\right)}{2m\omega_{\mathrm{n}k} \sqrt{ \xi_k^2-1}}\ \ \mathrm{if}\ \xi_k>1\mathrm{,}
    \end{aligned}
    \right.
\end{align*}
with 
\begin{align}\label{eq:omega_n-xi}
\omega_{\mathrm{n}k}:=\sqrt{\frac{\lambda_k}{m}}\qquad\text{and}\qquad\xi_k:=\frac{d}{2\sqrt{\lambda_k m}}\,.
\end{align}
\end{lem}
\begin{proof}
See the Appendix \ref{app:lem4-pf}.
\end{proof}

Consider an arbitrary Laplacian matrix of the network $\boldsymbol{L}_\mathrm{B}\in \real^{n \times n}$. For the case $\boldsymbol{R}=\boldsymbol{I}_n$, we have $\boldsymbol{L}=\boldsymbol{L}_\mathrm{B}$ by \eqref{eq:L-def}. From the definition of eigenvalues of $\boldsymbol{L}$, i.e., $\boldsymbol{L}\boldsymbol{v}_k=\lambda_k \boldsymbol{v}_k$, it is easy to see that, if all weights of $\boldsymbol{L}_\mathrm{B}$ are scaled by some $\beta>0$, the eigenvalues of $\boldsymbol{L}$ will also be scaled by $\beta$ since $(\beta\boldsymbol{L})\boldsymbol{v}_k=\beta(\boldsymbol{L}\boldsymbol{v}_k)=\beta(\lambda_k \boldsymbol{v}_k)=(\beta\lambda_k) \boldsymbol{v}_k$.
This indicates that, for the same network topology, larger weights of $\boldsymbol{L}_\mathrm{B}$ yield larger $\lambda_k$. 
Lemma~\ref{lem:hk-exp} suggests that the particular form of each $h_k(t)$ depends on the value of damping ratio $\xi_k$ which is determined by $\lambda_k$ via \eqref{eq:omega_n-xi}. Thus, even for the same network topology, the worst-case frequency nadir could vary significantly due to the changes in $h_k(t)$.  

\subsection{Power Networks with Strong Connectivity}
Here we show that if weights of $\boldsymbol{L}_\mathrm{B}$ are large enough, then the worst-case frequency nadir under power disturbances with a given $2$-norm bound occurs when the disturbance is evenly injected at all buses. Interestingly, this result holds for all network topologies. 

\begin{thm}[Worstness of even disturbances in networks with strong connectivity]\label{thm:even-largeweights}
 Under Assumption~\ref{ass:proportion} with $\boldsymbol{R}=\boldsymbol{I}_n$ ($n\geq 2$) and the representative generation unit given by \eqref{eq:go}, if 
\begin{align}\label{eq:connect-strong}
    \lambda_2\geq\left(n-0.75\right)d^2/m\,,
\end{align}
then, under the $2$-norm constraint on power disturbances, i.e., $\|\boldsymbol{u}_0\|_2\leq\rho$, the worst-case frequency nadir is $\rho/\left(d\sqrt{n}\right)$, achieved at $\boldsymbol{u}_0= \pm\rho\mathbbold{1}_{n}/\sqrt{n}$.
\end{thm}
\begin{proof} 
We prove the result by first showing that the frequency nadir on an arbitrary bus has the upper bound $\rho/\left(d\sqrt{n}\right)$ when \eqref{eq:connect-strong} holds and then showing that this upper bound can be achieved at $\boldsymbol{u}_0= \pm\rho\mathbbold{1}_{n}/\sqrt{n}$.

First note that the condition \eqref{eq:connect-strong} actually implies that, 
\begin{align}\label{eq:lambda_k-bound}
\forall k>1,\quad\lambda_k\geq\left(n-0.75\right)d^2/m>0.25d^2/m\,,   
\end{align}
where the first inequality is due to the fact that $\lambda_2$ is the second smallest eigenvalue of $\boldsymbol{L}$ and the second inequality holds if there is more than one bus in the network. Then, a simple calculation shows that \eqref{eq:lambda_k-bound} is equivalent to 
\begin{align}\label{eq:xik-bound}
\forall k>1 \,,\quad    1>\frac{d}{2\sqrt{\lambda_k m}}\,.
\end{align}
Thus, by the definition of damping ratio $\xi_k$ in \eqref{eq:omega_n-xi}, \eqref{eq:xik-bound} ensures that, $\forall k>1$, $0\leq\xi_k<1$, which allows us to bound $|h_k(t)|$, $\forall k>1$, as follows:
\begin{align}\label{eq:hk-bound}
    &|h_k(t)|\nonumber\\\stackrel{\circled{1}}{\leq}&\dfrac{e^{-\xi_k\omega_{\mathrm{n}k} t}}{m\omega_{\mathrm{n}k} \sqrt{1 - \xi_k^2}}\stackrel{\circled{2}}{=}\dfrac{e^{-\frac{d}{2\sqrt{\lambda_k m}}\sqrt{\frac{\lambda_k}{m}} t}}{m\sqrt{\frac{\lambda_k}{m}} \sqrt{1 - \left(\frac{d}{2\sqrt{\lambda_k m}}\right)^2}}\nonumber\\
    =&\ \!\dfrac{e^{-\frac{d}{2m}t}}{ \sqrt{\lambda_km \!-\! \frac{d^2}{4}}}\!\stackrel{\circled{3}}{\leq}\!\dfrac{e^{-\frac{d}{2m}t}}{ \sqrt{\left(n-0.75\right)\frac{d^2}{m}m \!-\! \frac{d^2}{4}}}\!=\!\dfrac{e^{-\frac{d}{2m}t}}{ d\sqrt{n-1}}\,,
\end{align}
where $\circled{1}$ uses the expression of $|h_k(t)|$ for the case when $0\leq\xi_k<1$ provided by Lemma~\ref{lem:hk-exp} and the fact that $|\sin(\cdot)|\leq1$, $\circled{2}$ uses the definitions of $\omega_{\mathrm{n}k}$ and $\xi_k$ in \eqref{eq:omega_n-xi} in Lemma~\ref{lem:hk-exp}, and $\circled{3}$ uses the first inequality in \eqref{eq:lambda_k-bound}.

We are ready to bound the frequency nadir on an arbitrary bus under the 2-norm constraint, whose expression can be obtained from \eqref{eq:cen-opt} and \eqref{eq:sol-2norm} as: 
\begin{align*}
   \max_{t>0}\!\dfrac{\rho}{\sqrt{r_i}}\!\left\lVert\boldsymbol{R}^{-\frac{1}{2}}\!\boldsymbol{V}\!\boldsymbol{c}_i(t) \right\rVert _2\!\!\!=\!\max_{t>0}\rho\left\lVert\boldsymbol{V}\!\boldsymbol{c}_i(t) \right\rVert _2\!=\!\max_{t>0}\rho\left\lVert\boldsymbol{c}_i(t) \right\rVert _2\!\,,
\end{align*} 
where the first equality uses the assumption $\boldsymbol{R}=\boldsymbol{I}_n$ and the secondary equality uses the unitary invariance property of $2$-norm. Recall that $\boldsymbol{c}_i(t):=\left(v_{k,i}h_k(t), k \in \mathcal{N} \right) \in \real^n$. The frequency nadir on bus $i$, $\forall i \in\mathcal{N}$, can be bounded as follows:
\begin{align}
   &\max_{t>0}\rho\left\lVert\boldsymbol{c}_i(t) \right\rVert _2\nonumber\\=&\max_{t>0}\rho\sqrt{ \sum_{k=1}^nv^2_{k,i}h^2_k(t)}
   =\max_{t>0}\rho\sqrt{ v^2_{1,i}h^2_1(t)+\sum_{k=2}^nv^2_{k,i}h^2_k(t)}\nonumber\\
   \stackrel{\circled{1}}{\leq}&
   \max_{t>0}\rho\sqrt{ \left(\dfrac{1}{\sqrt{n}}\right)^2\!\!\dfrac{1}{d^2}\left(1-e^{-\frac{d}{m}t}\right)^2\!+\!\sum_{k=2}^nv^2_{k,i}\!\left(\dfrac{e^{-\frac{d}{2m}t}}{ d\sqrt{n-1}}\right)^2}\nonumber\\
   =&\max_{t>0}\dfrac{\rho}{d}\sqrt{ \dfrac{1-2e^{-\frac{d}{m}t}+e^{-\frac{2d}{m}t}}{n}+\dfrac{e^{-\frac{d}{m}t}}{n-1}\left(\sum_{k=2}^nv^2_{k,i}\right)}\nonumber\\
   \stackrel{\circled{2}}{=}&\max_{t>0}\dfrac{\rho}{d}\sqrt{ \dfrac{1-2e^{-\frac{d}{m}t}+e^{-\frac{2d}{m}t}}{n}+\dfrac{e^{-\frac{d}{m}t}}{n-1}\left(1-v^2_{1,i}\right)}\nonumber\\
   \stackrel{\circled{3}}{=}&\max_{t>0}\dfrac{\rho}{d}\sqrt{ \dfrac{1-2e^{-\frac{d}{m}t}+e^{-\frac{2d}{m}t}}{n}+\dfrac{e^{-\frac{d}{m}t}}{n-1}\left[1-\left(\dfrac{1}{\sqrt{n}}\right)^2\right]}\nonumber\\
   =&\dfrac{\rho}{d\sqrt{n}}\max_{t>0}\sqrt{ 1-e^{-\frac{d}{m}t}+e^{-\frac{2d}{m}t}}\,.\label{eq:Nadir-busi}
\end{align} 
In $\circled{1}$, the expression of $\boldsymbol{v}_1$ in \eqref{eq:Vmatrix} for the case when $\boldsymbol{R}=\boldsymbol{I}_n$, i.e., $\boldsymbol{v}_1= \mathbbold{1}_n/\sqrt{n}$, the expression of $h_1(t)$ in \eqref{eq:h1-def} in Lemma~\ref{lem:hk-exp}, and the bound for $|h_k(t)|$ in \eqref{eq:hk-bound} are used. In $\circled{2}$, $\sum_{k=1}^nv^2_{k,i}$ is the squared $2$-norm of the $i$th row of $\boldsymbol{V}$. 
In $\circled{3}$, $\boldsymbol{v}_1= \mathbbold{1}_n/\sqrt{n}$ is used again.

Now, the problem of finding an upper bound for frequency nadir on an arbitrary bus turns into the problem \eqref{eq:Nadir-busi}. It is trivial to see that $-1=-1+0\leq-e^{-\frac{d}{m}t}+e^{-\frac{2d}{m}t}\leq0$, $\forall t>0$, which implies that 
\begin{align*}
    0\leq\sqrt{ 1-e^{-\frac{d}{m}t}+e^{-\frac{2d}{m}t}}\leq1\,,\quad\forall t>0\,.
\end{align*}
Thus, we have
\begin{align}\label{eq:bound-max}
    \max_{t>0}\sqrt{ 1-e^{-\frac{d}{m}t}+e^{-\frac{2d}{m}t}}\leq1\,.
\end{align}
Substituting \eqref{eq:bound-max} to \eqref{eq:Nadir-busi} yields the following upper bound for frequency nadir on an arbitrary bus: $\max_{t>0}\rho\!\left\lVert\boldsymbol{c}_i(t) \right\rVert _2\leq\rho/\left(d\sqrt{n}\right)$.

Finally, we show that this upper bound can be achieved by $\boldsymbol{u}_0= \pm\rho\mathbbold{1}_{n}/\sqrt{n}=\pm\rho\boldsymbol{v}_1$. From~\eqref{eq:u0-vk}, it is easy to see that, for $\boldsymbol{u}_0=\pm\rho\boldsymbol{v}_1$ and $\boldsymbol{R}=\boldsymbol{I}_n$, we have $\alpha_1=\pm\rho$ and $\alpha_k=0$, $\forall k\in\mathcal{N}\setminus\{1\}$. Thus, by Lemma~\ref{lem:w-t} and \eqref{eq:h1-def}, we have 
\begin{align}\label{eq:omega-equal-dis}
 \boldsymbol{\omega}(t)=\pm\rho h_1(t)\boldsymbol{v}_1 =  \pm\dfrac{\rho}{d}\left(1-e^{-\frac{d}{m}t}\right)\dfrac{\mathbbold{1}_{n}}{\sqrt{n}}\,.
\end{align}
Clearly, all buses in \eqref{eq:omega-equal-dis} exhibit the same behavior whose magnitude of frequency deviation monotonically increases on $t>0$. Thus, the frequency nadir on each bus is $|\omega_i(\infty)|=\rho /\left(d\sqrt{n}\right)$, which equals the upper bound shown above.
\end{proof}
\begin{rem}[Heterogeneous power network with strong connectivity]
    For power networks with heterogeneous generation units, some conclusions can still be drawn by
analyzing perturbations from the proportional case~\cite[Section VI-A]{pm2019preprint}. For example, although one cannot give a clean lower bound on $\lambda_2$ of $\boldsymbol{L}$ like \eqref{eq:connect-strong} to characterize the strong enough connectivity ensuring that the worst-case frequency nadir occurs when the disturbance is evenly
injected at all buses, it is found that all buses tend to respond in the same way as
the scaled representative generation unit when heterogeneous power network connectivity grows~\cite[Section VI-A]{pm2019preprint}. This allows us to show that, as
$\lambda_2\to\infty$ in heterogeneous power networks, the worst-case frequency nadir under the $2$-norm constraint on power disturbances is $(\rho\sqrt{n})/(d\sum_{i=1}^n \!r_i)$, achieved at $\boldsymbol{u}_0= \pm\rho\mathbbold{1}_{n}/\sqrt{n}$, which is a counterpart to Theorem~\ref{thm:even-largeweights}.
\end{rem}
\subsection{Power Networks with Weaker Connectivity}\label{ssec:weak-homo}

In power networks where the condition \eqref{eq:connect-strong} does not hold, the analysis of the worst-case frequency nadir becomes intractable. Therefore, we provide some intuitive discussions on what happens in power networks with weaker connectivity instead. If we consider the ``extreme'' case where weights of $\boldsymbol{L}_\mathrm{B}$ are sufficiently small to the extent that they can be deemed as $0$, we end up with a power network composed of isolated buses. Clearly, the worst-case frequency nadir, in this case, happens when the allowed power disturbance is injected at any single bus, with a value of $\rho/d$, since the frequency dynamics on that particular bus is exactly $\rho h_1(t)$ given in \eqref{eq:h1-def}. 

Of course, in practice, networks have ``intermediate'' connectivity. This is where the worst-case disturbance is neither concentrated at a single bus nor spread out to all buses, which, however, can be efficiently solved by our proposed Algorithm~\ref{alg:Nadir}.

\section{Numerical Illustrations}\label{sec:simulation}
In this section, we provide numerical validation for the performance of our proposed worst-case frequency nadir computation algorithm on two standard benchmarks. 

\subsection{Case Study: $50$-Generator Dynamic Test Case}\label{ssec:50}
We first test our algorithm on a power network with $50$ generation units, modified from the $50$-generator dynamic test case available in the Power Systems Toolbox~\cite{chow1992toolbox}. The inertia coefficients $m_i$ of the generation units and the Laplacian elements $\boldsymbol{L}_{\mathrm{B},{ij}}$ of the power network are obtained from the dataset via standard calculations with minor modifications to $m_i$. Then, we define the representative generation inertia
as the mean of $m_i$, i.e., $m := (\sum_{i=1}^{|\mathcal{N}|} m_i)/|\mathcal{N}|=\SI{559.69}{\second}$, where $|\mathcal{N}|=50$. Accordingly, the proportionality parameters are given by $r_i := m_i/m$, $\forall i \in \mathcal{N}$. Given that the values of damping coefficients are not provided by the dataset, we set $d_i = r_i d$ with the representative generation unit damping $d=\SI{2044.52}{\pu}$.\footnote{All per unit values are on the system base, where the system power base
is $S_0=\SI{100}{\mega\VA}$ and the nominal system frequency is $F_0=\SI{50}{\hertz}$.}


In the simulation, we implement Algorithm~\ref{alg:Nadir} developed in Section~\ref{ssec:alg}. Specifically, we consider a $2$-norm constraint on power disturbances with $\rho=0.5$, i.e., $\|\boldsymbol{u}_0\|_2\leq0.5$. With the time steps chosen as $\delta t=\SI{0.01}{\second}$ and $N=100$, Algorithm~\ref{alg:Nadir} takes around $\SI{0.7}{\second}$ in Matlab on MacBook Pro personal laptop to identify the power disturbances $\boldsymbol{u}_0^\star$ shown in Fig.~\ref{fig:dist-comp-50} which will produce the worst-case frequency nadir. 

\begin{figure}[t!]
\centering
\subfigure[Power disturbances, where a larger edge length reflects weaker connectivity and a smaller node size denotes lower inertia]
{\includegraphics[width=\columnwidth]{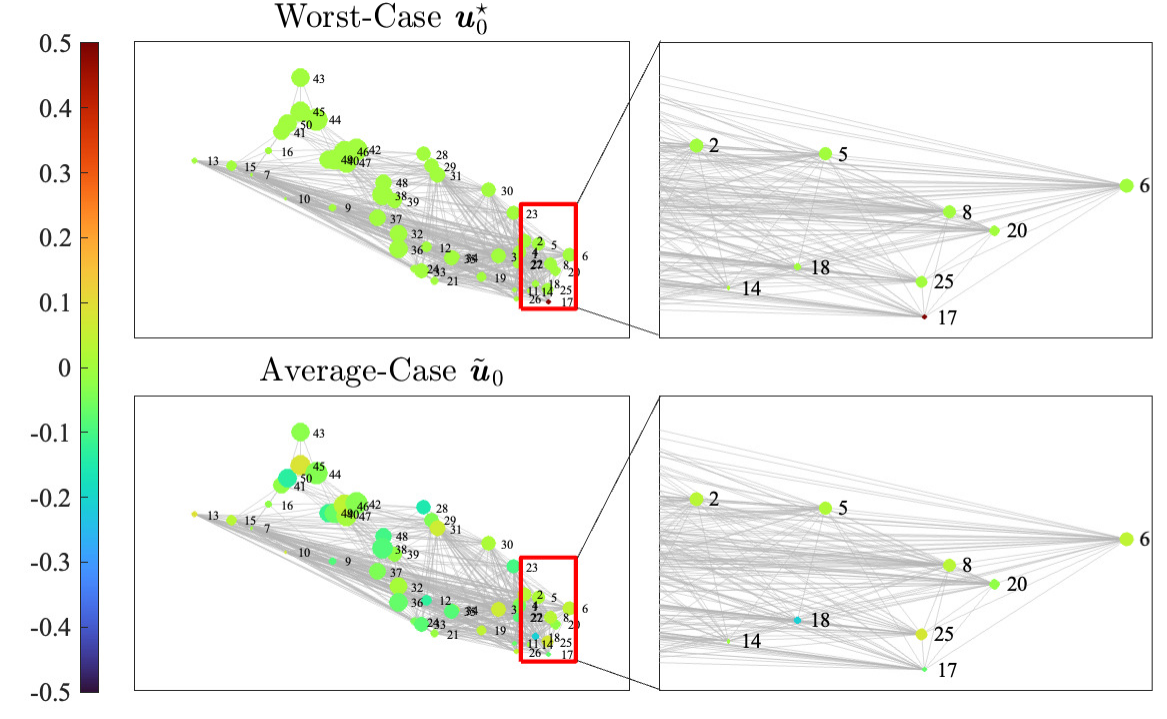}\label{fig:dist-comp-50}}
\hfil
\subfigure[Frequency deviations, where the legends provide the index of the bus with the maximum frequency nadir in each case]
{\includegraphics[width=\columnwidth]{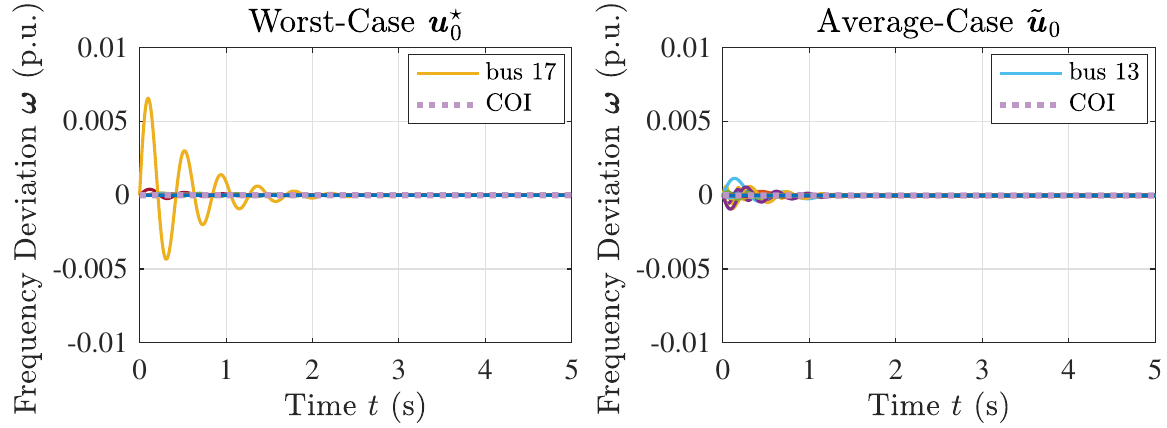}\label{fig:fre-comp-50}}
\caption{Comparison between the worst and a random allocation of power disturbances within a given $2$-norm bound as well as the resultant frequency deviations in a power network with $50$ generation units.}
\label{fig:compare-50}
\end{figure}

It is instructive to compare the worst-case disturbance $\boldsymbol{u}_0^\star$ and the ``average-case'' disturbance  $\tilde{\boldsymbol{u}}_0$ in this network, as shown in Fig.~\ref{fig:fre-comp-50}, where $\tilde{\boldsymbol{u}}_0$ is obtained by sampling random disturbances with $\|\tilde{\boldsymbol{u}}_0\|_2=0.5$.\footnote{This is done by first drawing each element of $\tilde{\boldsymbol{u}}_0$ from
the standard normal distribution and then scaling the vector to
have the desired $2$-norm.} Fig.~\ref{fig:fre-comp-50} confirms that the frequency nadir triggered by the worst-case disturbance is much larger than the one seen in an average case. In addition, it is not obvious why bus $17$ has the worst-case frequency nadir since it is neither the point with the weakest connectivity (bus $13$) nor the point with the lowest inertia (bus $10$). This shows that explicitly computing the frequency behaviors is useful. 

\subsection{Case Study: $2000$-Bus Great Britain Network}
We also test our algorithm on a power network with $394$ generation units, based on the $2000$-bus Great Britain network taken from the Power Systems Test Case Archive.  The Laplacian elements $\boldsymbol{L}_{\mathrm{B},{ij}}$ of the power network are obtained from the dataset via standard calculations. Since inertia and damping information are not contained in the original dataset, we draw $m_i$ uniformly at random from $(0,1000)\SI{}{\second}$, yielding $m:= (\sum_{i=1}^{|\mathcal{N}|} m_i)/|\mathcal{N}|=\SI{497.18}{\second}$ and $r_i := m_i/m$, $\forall i\in \mathcal{N}$, and then we set $d_i = r_i d$ with the representative generation unit damping $d=\SI{1816.18}{\pu}$. Other simulation settings are the same as in Section~\ref{ssec:50}, i.e., $\rho=0.5$, $\delta t=\SI{0.01}{\second}$, and $N=100$.

For this setting, Algorithm~\ref{alg:Nadir} takes around $\SI{224}{\second}$ to identify that the worst-case $\boldsymbol{u}_0^\star$ will produce the $\SI{0.0075}{\pu}$ frequency nadir, occurring at $t=\SI{0.46}{\second}$ on bus $364$, which is confirmed by the plot of frequency deviations in Fig.~\ref{fig:compare-2000} obtained from the dynamic simulation. 

\begin{figure}[t!]
\centering
\includegraphics[width=\columnwidth]{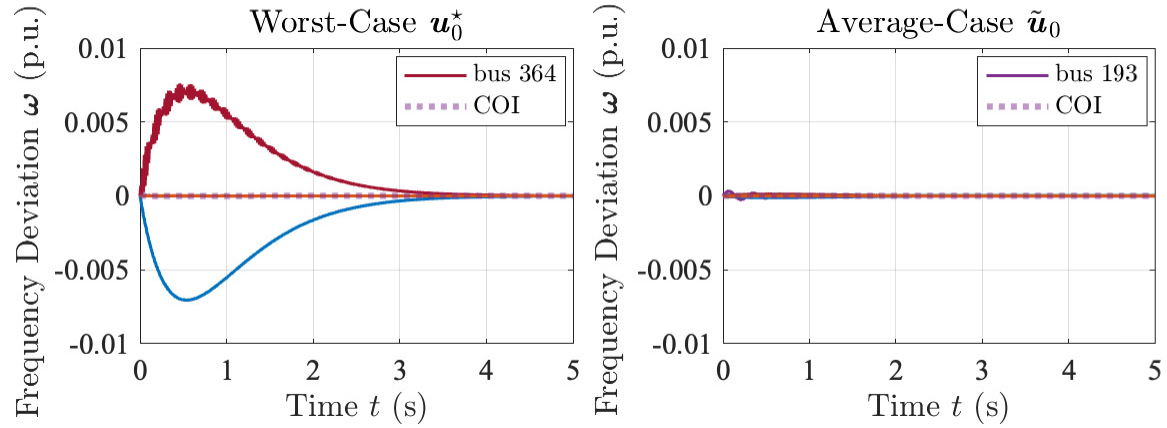}
\caption{Comparison between the frequency deviations in the Great Britain network under the worst and a random allocation of power disturbances within a given $2$-norm bound. In this particular example, there are two buses (buses $364$ and $365$) that swing each other in the worst-case.}
\label{fig:compare-2000}
\end{figure}

\section{Conclusions and Outlook}\label{sec:conclusion}

Above all, we have presented an efficient algorithm for computing the worst-case frequency nadir under a certain norm bound on power disturbances that the power network should survive. Compared to the COI frequency, the worst-case frequency nadir provided by such an algorithm is especially reliable for providing frequency security assessment to power grids with increased amount of renewable resources located in remote areas entering through long transmission lines, which are known to suffer from oscillations. In addition, we have demonstrated the impact of network connectivity on the worst-case frequency nadir in power networks with homogeneous generation units. Our future research will concentrate on some important extensions, including analysis of the worst-case frequency nadir for power networks with intermediate connectivity and heterogeneous parameters.

\appendix[]

\subsection{Proof of Lemma~\ref{lem:w-t}} \label{app:lem1-pf}
First, \eqref{eq:alpha} is due to \eqref{eq:u0-vk} since $\boldsymbol{\alpha}=(\boldsymbol{R}^{\frac{1}{2}}\boldsymbol{V})^{-1}\boldsymbol{u}_0=\boldsymbol{V}^T\boldsymbol{R}^{-\frac{1}{2}}\boldsymbol{u}_0$. Now, 
substituting $\hat{\boldsymbol{p}}(s)=\boldsymbol{u}_0/s$ to \eqref{eq:Tp} yields
 \begin{align}
\hat {\boldsymbol{\omega}}(s) =\ & \boldsymbol{R}^{-\frac{1}{2}} \boldsymbol{V} \diag{\hat{z}_{k}(s), k \in \mathcal{N}} \boldsymbol{V}^T \boldsymbol{R}^{-\frac{1}{2}}\boldsymbol{u}_0/s\nonumber
\\ \stackrel{\circled{1}}{=}\ & \boldsymbol{R}^{-\frac{1}{2}} \boldsymbol{V} \diag{\hat{z}_{k}(s), k \in \mathcal{N}} \boldsymbol{\alpha}/s\nonumber\\
=\ & \boldsymbol{R}^{-\frac{1}{2}} \boldsymbol{V} \sum_{k=1}^n  \dfrac{\alpha_k\hat{z}_{k}(s)}{s}\boldsymbol{e}_k=\boldsymbol{R}^{-\frac{1}{2}} \sum_{k=1}^n  \alpha_k\dfrac{\hat{z}_{k}(s)}{s}\boldsymbol{V} \boldsymbol{e}_k\nonumber\\\stackrel{\circled{2}}{=}\ &\boldsymbol{R}^{-\frac{1}{2}} \sum_{k=1}^n \alpha_k \dfrac{\hat{z}_{k}(s)}{s} \boldsymbol{v}_k\nonumber\\\stackrel{\circled{3}}{=} \ &\sum_{k=1}^n  \alpha_k\frac{\hat{g}_\mathrm{o}(s)}{s+\lambda_k\hat{g}_\mathrm{o}(s)} \boldsymbol{R}^{-\frac{1}{2}}\boldsymbol{v}_k\;,\label{eq:w(s)-sum}
\end{align} 
where \circled{1} uses \eqref{eq:alpha}, \circled{2} uses $\boldsymbol{e}_k$ to plug out the $k$th column of $\boldsymbol{V}$, and \circled{3} uses \eqref{eq:zk-s}. Then, \eqref{eq:omega-t-decomple} follows directly from taking inverse Laplace transform to \eqref{eq:w(s)-sum}. Finally, substituting the explicit expression of $\boldsymbol{v}_1$ in \eqref{eq:Vmatrix} to \eqref{eq:omega-t-decomple} yields \eqref{eq:omega-t-decomple-coi}.

\subsection{Proof of Lemma~\ref{lem:dual-norm}} \label{app:lem2-pf}
Part of the proof can be found in textbooks. Here, we provide a complete proof for self-containedness.

For the $2$-norm case, by Cauchy-Schwartz inequality~\cite[Theorem 5.1.4]{Horn2012MA} and the norm constraint, we have $\left\lvert\boldsymbol{x}^T\boldsymbol{y} \right\rvert \leq\|\boldsymbol{x}\|_2\|\boldsymbol{y} \|_2\leq\|\boldsymbol{x}\|_2$,
where the upper bound is clearly achieved at $\boldsymbol{y}^\star=\pm \boldsymbol{x}/\|\boldsymbol{x}\|_2$ since $\left\lvert\boldsymbol{x}^T(\pm \boldsymbol{x}/\|\boldsymbol{x}\|_2) \right\rvert=\left\lvert\pm \boldsymbol{x}^T\boldsymbol{x} \right\rvert/\|\boldsymbol{x}\|_2=\|\boldsymbol{x}\|_2^2/\|\boldsymbol{x}\|_2=\|\boldsymbol{x}\|_2$. It is easy to check that $\|\boldsymbol{y}^\star\|_2=1$.

For the $\infty$-norm case, we have $\left\lvert\boldsymbol{x}^T\boldsymbol{y} \right\rvert =\left\lvert\sum_{i=1}^nx_iy_i\right\rvert\leq\sum_{i=1}^n|x_i||y_i|\leq(\max_{i\in\mathcal{N}}|y_i|)\sum_{i=1}^n|x_i|=:\|\boldsymbol{y}\|_\infty\|\boldsymbol{x}\|_1\leq\|\boldsymbol{x}\|_1$, where the last inequality uses the norm constraint. This upper bound is achieved at $\boldsymbol{y}^\star=\pm\mathrm{sign}(\boldsymbol{x})$ since $\left\lvert\boldsymbol{x}^T(\pm \mathrm{sign}(\boldsymbol{x}))\right\rvert =\left\lvert\pm\boldsymbol{x}^T\mathrm{sign}(\boldsymbol{x}) \right\rvert=\left\lvert\boldsymbol{x}^T\mathrm{sign}(\boldsymbol{x}) \right\rvert=\sum_{i=1}^n|x_i|=:\|\boldsymbol{x}\|_1$. It is easy to check that $\|\boldsymbol{y}^\star\|_\infty=1$.

For the $1$-norm case, we have $\left\lvert\boldsymbol{x}^T\boldsymbol{y} \right\rvert =\left\lvert\sum_{i=1}^nx_iy_i\right\rvert\leq\sum_{i=1}^n|x_i||y_i|\leq(\max_{i\in\mathcal{N}}|x_i|)\sum_{i=1}^n|y_i|=:\|\boldsymbol{x}\|_\infty\|\boldsymbol{y}\|_1\leq\|\boldsymbol{x}\|_\infty$, where the last inequality uses the norm constraint. This upper bound is achieved at $\boldsymbol{y}^\star=\pm\boldsymbol{e}_{i_\mathrm{m}}$ with $i_\mathrm{m}:=\arg\max_{i\in\mathcal{N}}|x_i|$ since $\left\lvert\boldsymbol{x}^T(\pm \boldsymbol{e}_{i_\mathrm{m}} )\right\rvert =\left\lvert\boldsymbol{x}^T\boldsymbol{e}_{i_\mathrm{m}} \right\rvert=\left\lvert x_{i_\mathrm{m}}\right\rvert=\max_{i\in\mathcal{N}}|x_i|=:\|\boldsymbol{x}\|_\infty$. It is easy to check that $\|\boldsymbol{y}^\star\|_1=1$.

\subsection{Proof of Lemma~\ref{lem:hk-exp}} \label{app:lem4-pf}
Substituting \eqref{eq:go} to \eqref{eq:hk-def} yields
\begin{align*}
    h_1(t)\!=&\!\ \mathscr{L}^{-1}\!\left\{\frac{1}{s\left(ms+d\right)}\right\}\!=\!\mathscr{L}^{-1}\!\left\{\dfrac{1}{d}\!\left(\dfrac{1}{s}-\frac{1}{s+d/m}\right)\right\}\,,\nonumber \\
    h_k(t)=&\!\ \mathscr{L}^{-1}\!\left\{\!\frac{1}{m\left(s^2+2\xi_k\omega_{\mathrm{n}k}s+\omega_{\mathrm{n}k}^2\right)}\!\right\}\,, \quad\forall k>1\,,\nonumber
\end{align*}
with the natural frequency $\omega_{\mathrm{n}k}$ and damping ratio $\xi_k$ defined by \eqref{eq:omega_n-xi}. Note that the explicit expression  of $h_k(t)$, $\forall k>1$, depends on the value of damping ratio $\xi_k$ as follows:
\begin{enumerate}
    \item Under damped case ($0\leq\xi_k<1$): 
    \begin{align}
    h_k(t)\!=\!\mathscr{L}^{-1}\!\!\left\{\!\dfrac{1}{m\omega_{\mathrm{n}k} \sqrt{1 \!-\! \xi_k^2}}\dfrac{\omega_{\mathrm{n}k} \sqrt{1 - \xi_k^2}}{\left(s\!+\!\xi_k\omega_{\mathrm{n}k}\right)^2\!+\!\omega_{\mathrm{n}k}^2\!\left(1 \!-\! \xi_k^2\right)}\!\right\}\!\,;\nonumber
\end{align}
\item Critically damped case ($\xi_k=1$):
\begin{align*}
   h_k(t)\!=&\ \!\mathscr{L}^{-1}\!\left\{\!\dfrac{1}{m}\dfrac{1}{\left(s+\omega_{\mathrm{n}k}\right)^2}\!\right\}\,;
\end{align*}
\item Over damped case ($\xi_k>1$):
\begin{align*}
   h_k(t)\!=&\ \!\mathscr{L}^{-1}\!\left\{\!\dfrac{1}{m\left(\sigma_{1,k}-\sigma_{2,k}\right)}\!\left(\dfrac{1}{s+\sigma_{2,k}}-\dfrac{1}{s+\sigma_{1,k}}\right)\!\right\}
\end{align*}
with $\sigma_{1,k} := \omega_{\mathrm{n},k}\!\left(\xi_k\!+\! \sqrt{\xi_k^2\!-\!1}\right)$ and $\sigma_{2,k} := \omega_{\mathrm{n},k}\left(\xi_k\!-\! \sqrt{\xi_k^2\!-\!1}\right)$. 
\end{enumerate}
Therefore, the results follow directly from the standard inverse Laplace transform formulas.

\bibliographystyle{IEEEtran}
\bibliography{main}

\end{document}